\documentclass[twoside,leqno,twocolumn]{article}
\usepackage{ltexpprt}
\usepackage{algorithm}
\usepackage{algorithmic}
\usepackage{amssymb}
\usepackage{ifpdf}
\usepackage[cmex10]{amsmath}
\usepackage{array}
\usepackage{times}
\ifpdf
  \usepackage[pdftex]{graphicx}
  \DeclareGraphicsExtensions{.pdf,.jpg,.png}
\else
  \usepackage[dvips]{graphicx}
  \graphicspath{{./figures/}}
  \DeclareGraphicsExtensions{.eps}
\fi
\usepackage{subfigure}

\def\LDAG{{\it LDAG}}
\def\IBS{{\it IBS}}
\def\NIR{{\it NIR}}
\def\Inf{{\it Inf}}
\def\mequal{\,\mathrm{-\hspace{-1mm}=}\ }
\def\pequal{\,\mathrm{+\hspace{-1mm}=}\ }
\def\ap{{\it ap}}

\title{Influence Blocking Maximization in Social Networks \\
	under the Competitive Linear Threshold Model \ \\\ \\ Technical Report}
\author{Xinran He$^\dag$\ \ \ Guojie Song$^\dag$\ \ \ Wei Chen$^\ddag$
	\ \ \ Qingye Jiang$^\S$ \\
$^\dag$ Ministry of Education Key Laboratory of Machine Perception,
	Peking University, China \\
	\tt{\{xinranhe,gjsong\}@pku.edu.cn} \\
$^\ddag$ Microsoft Research Asia, China, \tt{weic@microsoft.com} \\
$^\S$ Columbia University, U.S.A., \tt{qj2116@columbia.edu}
}

\begin{document}
\date{}
\maketitle

\begin{abstract}
\small\baselineskip=9pt
In many real-world situations, different and often opposite opinions, innovations, or
	products are competing with one another for their social influence in
	a networked society.
In this paper, we study competitive influence propagation in
	social networks under the
	 competitive linear threshold (CLT) model, an extension to the classic
	linear threshold model.
Under the CLT model, we focus on the problem that one entity
	tries to block the influence propagation of its competing entity as
	much as possible by strategically selecting a number of seed nodes
	that could initiate its own influence propagation.
We call this problem the influence blocking maximization (IBM) problem.
We prove that the objective function of IBM in the CLT
	model is submodular,
	and thus a greedy algorithm could achieve $1-1/e$ approximation
	ratio.
However, the greedy algorithm requires Monte-Carlo simulations of	
	competitive influence propagation, which makes the algorithm not
	efficient.
We design an efficient algorithm CLDAG, which utilizes the properties of
	the CLT model, to address this issue.
We conduct extensive simulations of CLDAG, the greedy algorithm, and other
	baseline algorithms on real-world and synthetic datasets.
Our results show that CLDAG is able to provide best accuracy in par with
	the greedy algorithm and often better than other algorithms, while
	it is two orders of magnitude faster than the greedy algorithm.

\vspace{.5mm}

\end{abstract}

\section{Introduction}
With the increasing popularity of online social and information
	networks such as Facebook, Twitter, LinkedIn, etc., many
	researchers have studied diffusion phenomenon in social networks,
	which includes the diffusion of news, ideas, innovations, adoption
	of new products, etc.
We generally refer to such diffusions as influence diffusion or propagation.
One topic in influence diffusion that has been extensively studied is
	{\em influence
	maximization}~\cite{KKT03,KKT05,KS06,NN08,CWY09,CWW10,WCSX10,CYZ10}.
Influence maximization is the problem of selecting a small set of
	{\em seed nodes} in a social network, such that its overall influence
	coverage is maximized, under certain influence diffusion models.
Popular influence diffusion models include the independent cascade (IC)
	model and the linear threshold (LT) model, which was first summarized
	by Kempe et al. in~\cite{KKT03} based on prior research in social network
	analysis and particle physics.
Both IC and LT models are stochastic models characterizing how influence
	are propagated throughout the network starting from the initial seed nodes.

However, all of the above research works only study the diffusion
	of a single idea in the social networks.
In reality, it is often the case that different and often opposite
	information, ideas and innovations are competing for their influence
	in the social networks.
Such competing influence diffusion could range from two competing
	companies engaging in two marketing campaigns trying to grab people's
	attentions, or two political candidates of the opposing parties trying
	to influence their voters, to government authorities trying to inject
	truth information to fight with rumors spreading in the public, and so on.

Motivated by the above scenarios, several recent studies have
	looked into competitive influence diffusion and its corresponding
	influence maximization
	problems~\cite{BKS07,KJY08,PBS10,TTK10,BFO10,WWW11,Chen2011}.
Most of them propose some extensions to the existing influence diffusion
	models to incorporate competitive influence diffusion, and then either
	focus on the influence maximization problem for one of the competing
	parties, or study the game theoretic aspects of competitive influence
	diffusion (see Section~\ref{sec:related} for more details on these
	related works).
In this paper, we concentrate on the problem of how to block the
	influence diffusion of an opposing party as much as possible.
For example, when there is a negative rumor spreading in the social
	network about a company,
	the company may want to react quickly by
	selecting seed nodes to inject positive
	opinions about the company to fight against the negative rumor.
Similar situations could occur when a political candidate tries to
	stop a negative rumor about him or her, or when government or
	public officials try to stop erroneous rumors about public health and
	safety, terrorist threat, etc.
We call the problem of selecting positive seed nodes in a social network
	to minimize the effect of negative influence diffusion, or to
	maximize the blocking effect on negative influence, the
	{\em influence blocking maximization (IBM)} problem.

We study the IBM problem under a competitive linear threshold (CLT) model,
	which we extend naturally
	from the classic linear threshold model and is similar
	to a model proposed independently in~\cite{BFO10}.
We prove that the objective function of IBM under the CLT model is
	monotone and submodular, which means a standard greedy algorithm can
	achieve an approximation ratio of $1-1/e-\epsilon$ to the optimal solution,
	where $\epsilon$ is any positive number.
However, the greedy algorithm requires Monte-Carlo simulations of
	competitive influence diffusion, which becomes very slow for large
	networks, if we want to keep $\epsilon$ above small.
For example, in our simulation, for a network with 6.4k nodes, the greedy algorithm
	takes more than 8 hours to finish.
This is especially problematic for the IBM problem, since blocking influence
	diffusion usually requires very swift decisions before the negative
	influence propagates too far.
To address the efficiency issue, we utilize the efficient computation
	property of the LT model for directed acyclic graphs (DAGs), and design
	an efficient heuristic CLDAG for the IBM problem under the CLT model.
Because of the complex interaction in the competitive influence diffusion
	under the CLT model, we need a carefully designed dynamic programming
	method for influence computation in our CLDAG algorithm.
To test the efficiency and effectiveness of our CLDAG algorithm, we conduct
	extensive simulations on three real-world networks as well as synthetic
	networks.
We compare the performance of CLDAG with the greedy algorithm and other
	heuristic algorithms.
Our results show that (a) comparing with the greedy algorithm, our
	CLDAG algorithm achieves matching influence blocking effect while it runs
	two orders of magnitude faster; and (b) comparing with other
	heuristics such as degree-based heuristics, our algorithm consistently
	performs well and is often better than the other heuristics with a
	significant margin.

To the best of our knowledge, our work is the first to study the IBM problem under the competitive linear threshold
	model.
The study closest to our work is the one in~\cite{WWW11}, but they study
	the IBM problem under an extension of the independent cascade model,
	and due to the issue of non-submodularity, their study only works
	for a restricted extention to the IC model that is less natural.
Moreover, their work does not address the efficiency issue, which is
	vital to influence blocking maximization.

The rest of the paper is organized as follows.
We discuss related works in Section~\ref{sec:related}.
In Section~\ref{sec:model}, we specify the competitive linear threshold model.
In Section~\ref{sec:ibm}, we define the influence blocking maximization
	problem, show that it is NP-hard, and prove its submodularity
	under the CLT model.
We describe our CLDAG algorithm in Section~\ref{sec:algo}, and then provide
	our experimental evaluation results in Section~\ref{sec:exp}.
We conclude the paper with discussions in Section~\ref{sec:conclude}.

\section{Related Work} \label{sec:related}
Independent cascade model and linear
	threshold model are two extensively studied
	influence diffusions models originally
	summarized by Kempe et al.~\cite{KKT03}, based on earlier
	works of~\cite{G78,T78,JB01}.
Kempe et al. prove that the generalized versions of these two models
	are equivalent~\cite{KKT03}.
Based on the IC and LT model, Kempe et.al
~\cite{KKT03,KKT05} propose a greedy algorithm to solve the influence
maximization problem (brought about by Richardson~\cite{DR02}) to
maximize the spreading of a single piece of ideas, innovations, etc.
	under these two models.
Many follow-up studies propose alternative heuristics and try to
	solve the influence maximization problem more
	efficiently~\cite{KS06,NN08,CWY09,CWW10,CYZ10,WCSX10}.
In terms of efficient algorithm design, our work follows the idea
	in~\cite{CWW10,CYZ10} of finding efficient local graph structures to
	speed up the computation.
In particular, our CLDAG algorithm is similar to the LDAG algorithm
	of~\cite{CYZ10}, which is also based on the DAG structure, but
	our CLDAG algorithm is novel in dealing with competitive influence
	diffusion using the dynamic programming method.

Recently, there are a number of studies on competitive influence
	diffusion~\cite{BKS07,KJY08,PBS10,TTK10,BFO10,WWW11,Chen2011}.
Bharathi et al, extend the
	IC model to model competitive influence~\cite{BKS07},
	but they only provide a polynomial approximation algorithm for trees.
Kostka et al. study competitive rumor spreading~\cite{KJY08} on
	a more restricted model than IC and LT, and focused on showing the
	hardness of computing the optimal solution for the two competing parties.
Pathak et al. study a model of multiple cascades~\cite{PBS10}, which is
	an extension of a different
	diffusion model called the voter model~\cite{CS73,HL75}, even though
	they claim it to be a generalization of the linear threshold model.
They only study model dynamics and do not address the influence maximization
	problem.
Trpevski et al.~\cite{TTK10} propose another competitive rumor spreading
	model based on the epidemic model of SIS and study the dynamics in
	several classes of graphs, and they do not address the issue of influence
	maximization either.
Borodin et.al \cite{BFO10} extend the LT model in several different
	ways to model competitive influence diffusion, one of which
	is essentially our CLT model except for a different tie-breaking rule.
However, they only study the influence maximization problem,
	not the influence blocking maximization.
In particular, they show that influence maximization in the CLT model is
	not submodular, which is an interesting contrast to our result that
	influence
	blocking maximization under the CLT model {\em is} submodular.
We provide some reason in Section~\ref{sec:conclude} on why there is such
	a subtle difference.
The work of Budak et al.~\cite{WWW11} is the only one we found that studies
	influence blocking maximization (they call it eventual influence
	limitation), but they study this problem under an extension of the IC
	model.
They show that the general extension of the IC model in which positive
	influence and negative influence has a separate set of parameters (same
	as the case in our CLT model) is not submodular, and thus to achieve
	submodularity they have
	to restrict the model such that positive propagation probability is
	$1$ or is the same as negative propagation probability,
	which limits the expressiveness of the model.
Moreover, they only study the greedy algorithm and some simple heuristics,
	and do not provide efficient and scalable solution that
	maintains good accuracy at the same time.
Finally the work of~\cite{Chen2011} studies negative opinions emerging from
	poor product or service qualities, not generated by competitors.
They  study positive influence maximization under
	an extension to the IC model,
	and thus different from our study on blocking negative
	influence under the extension of the LT model.
The efficient influence maximization algorithm in~\cite{Chen2011} also uses
	dynamic programming, which bears some resemblance to our
	work.

\section{Competitive Linear Threshold Model} \label{sec:model}
Kempe et al. proposed the linear threshold model in \cite{KKT03} as a stochastic
model to address information cascade in a network. In
this model, a social network is considered as a directed graph
$G=(V,E)$, where $V$ is the set of vertices representing individuals
	and $E$ is the set of directed
	edges representing influence relationships among individuals.
Each edge $(u,v)\in E$ has a weight $w_{uv}\ge 0$, indicating the
	importance of $u$ in influencing $v$.
For convenience, for any $(u,v)\not\in E$, $w_{uv} = 0$.
For each $v \in V$, we have $\sum_{u\in V} w_{uv} \le 1$.
Each vertex $v$  picks an independent threshold $\theta_v$ uniformly at random
	from $[0,1]$.
Each vertex is either {\em inactive} or {\em active}, and once it is active, it
	stays active forever.
The diffusion process unfolds in discrete time steps.
At step $0$ a seed set $S\subseteq V$ is activated while all other
	vertices are inactive.
At any later step $t>0$, a vertex $v$ is activated if and only if
	the total weight of its active in-neighbors exceeds its threshold
	$\theta_v$, that is $\sum_{u\in S_{t-1}}w_{uv} \ge \theta_v$, where
	$S_{t-1} \subseteq V$ is the set of active vertices by time $t-1$,
	with $S_0 = S$.

We now extend the LT model to incorporate competitive influence diffusion.
The idea is that we allow each vertex to be positively activated or
	negatively activated, each of which is determined by concurrent
	positive diffusion and negative diffusion, respectively.
In the case that a vertex is both positively activated and negatively
	activated in the same step, then negative activation dominates the result.

More precisely, we define competitive linear threshold (CLT) model as
	an extension to the LT model in the following way.
Each vertex has three states, {\em inactive}, {\em +active}, and {\em
	\mbox{-active}}, and it does not change state once it becomes
	+active or -active.
Each edge $(u,v)$ has two weights, positive weight $w_{uv}^+$ and
	negative weight $w_{uv}^-$.
We can also think of it as $(u,v)$ splitting into
	two virtual edges, one positive
	edge propagating positive influence and one negative edge propagating
	negative influence.
Each vertex $v$ picks two independent
	thresholds uniformly at random from $[0,1]$,
	one positive threshold $\theta^+_v$ and one negative threshold
	$\theta^-_v$.
At step $0$, there are two disjoint seed sets, the positive
	seed set $P_0$ and the negative seed set $N_0$.
At each step $t$, positive influence and negative influence propagate	
	independently as in the original LT model, using positive weights/thresholds
	and negative weights/thresholds, respectively.
If a vertex $v$ is activated only by positive diffusion (or resp.
	negative diffusion), then $v$ becomes +active (or resp. -active).
If in step $t$ $v$ is activated by both positive diffusion and negative
	diffusion, then negative diffusion dominates and $v$ becomes
	-active.
The negative dominance rule reflects the negativity bias phenomenon well
	studied in social psychology, and matches the common sense that rumors
	are usually hard to fight with.

The CLT model defined here
	is essentially the same as the separate threshold model
	of~\cite{BFO10}, except that we use the negative dominance as the
	tie-breaking rule, while they use the random rule --- +active
	and -active status are picked uniformly at random.
We comment that the difference in the tie-breaking rule is not essential
	for our study: the submodularity property still holds
	and our algorithm can be properly adapted for the random tie-breaking
	rule.

\section{Influence Blocking Maximization Problem} \label{sec:ibm}

In this section, we first define the influence blocking
	maximization (IBM) problem, then show that IBM under the
	CLT model is NP-hard, and finally prove that the objective
	function of IBM is monotone and submodular, which leads to a
	greedy approximation algorithm.

\subsection{Problem definition.}

Informally, the IBM problem is an optimization problem in which
	given a graph $G=(V,E)$, its positive and negative edge weights, a
	negative seed set $N_0$, and a positive integer $k$, we want to
	find a positive seed set $S$ of size at most $k$ such that the
	expected number of negatively activated nodes is minimized, or
	equivalently, the reduction in the number of negatively activated
	nodes is maximized.

More precisely, let $\vec{\theta}^+$ and $\vec{\theta}^-$
	be the vector of positive thresholds and negative thresholds, respectively,
	for all vertices in $G$.
According to the CLT model, they are drawn from the probability space
	$[0,1]^{|V|}$ uniformly at random.
When $\vec{\theta}^+$ and $\vec{\theta}^-$ are fixed, all randomness in
	the CLT model is fixed.
Let $\IBS(S, N_0 \ |\ \vec{\theta}^+,\vec{\theta}^-)$ be the set of nodes
	$v$ in $G$ such that under thresholds $\vec{\theta}^+$ and $\vec{\theta}^-$,
	$v$ is negatively activated if $N_0$ is the negative seed set and
	positive seed set is empty,
	while $v$ is not negatively activated if $N_0$ is
	the negative seed set and $S$ is the positive seed set.
Thus this set represents the set of nodes that have been blocked from
	negative influence, and IBS stands for {\em influence blocking set}.
Since we always use $N_0$ as the negative seed set, we will omit
	$N_0$ from the notation for simplicity.
When the context is clear, we may also omit $\vec{\theta}^+$ and
	$\vec{\theta}^-$ and only use $\IBS(S)$ to represent the influence
	blocking set.
We define {\em negative influence reduction (NIR)} of a positive
	seed set $S$, denoted as $\sigma_{\NIR}(S)$, to be the expected
	value of the size of $\IBS(S\ |\ \vec{\theta}^+,\vec{\theta}^-)$,
	with expectation taken over all $\vec{\theta}^+$'s and $\vec{\theta}^-$'s,
	that is,
\[
\sigma_{\NIR}(S) = E_{\vec{\theta}^+,\vec{\theta}^-}
	(|\IBS(S\ |\ \vec{\theta}^+,\vec{\theta}^-)|).
\]

Then the {\em influence blocking maximization} is the problem of finding
	a positive seed set $S$ of size at most $k$ that maximizes
	$\sigma_{\NIR}(S)$, i.e., computing
\[
	P^{*}=\arg\max_{|S|\le k}\sigma_{\NIR}(S).
\]

We first show that the exact problem of IBM is NP-hard.

\begin{theorem}\label{thm:nph}
Under the CLT model, IBM problem is NP-hard.
\end{theorem}
\begin{proof}
By a reduction from the vertex cover problem. The full NP-hardness proof in presented in Appendix.
\hfill$\Box$
\end{proof}

\subsection{Submodularity of $\sigma_{\NIR}(S)$ and the
	greedy approximation algorithm.}
To overcome the NP-hardness result of Theorem~\ref{thm:nph}, we look
	for approximation algorithms.
The submodularity of set function $\sigma_{\NIR}(S)$ provides a good way
	to obtain an apporiximation algorithm for the IBM problem.
We say that a set function $f(S)$ with domain $2^V$ is {\em submodular} if
	for all $S\subseteq T \subseteq V$, and $x\not\in T$,
	we have $f(S\cup \{x\})-f(S) \ge f(T\cup \{x\})-f(T)$.
Intuitively, submodularity of $f$ means $f$ has the diminishing marginal
	return property.
Moreover, we say that $f$ is {\em monotone} if for all
	$S\subseteq T \subseteq V$, $f(S)\le f(T)$.

We now show that $\sigma_{\NIR}(S)$ is monotone and submodular.
We follow the general methodology as in~\cite{KKT03} for the proof, but
	our proof is more involved because of the complexity of our CLT model
	and the IBM problem.
We first construct an equivalent random process, and then use this
	random process to prove the result.

From the original graph $G=(V,E)$ with positive and negative weights,
	we construct a {\em random live-path graph} $G_X$ as follows.
For each $v\in V$, we randomly pick one positive
	in-edge $(u,v)$ with probability $w_{uv}^+$, and with
	probability $1-\sum_{u\in V}w_{uv}^{+}$ no positive in-edge is selected;
	we also randomly pick one negative
	in-edge $(u,v)$ with probability $w_{uv}^-$, and with
	probability $1-\sum_{u\in V}w_{uv}^{-}$ no negative in-edge is selected.
Let $G^+$ be the subgraph of $G_X$ consisting of only positive edges,
	and let $G^-$ be the subgraph of $G_X$ consisting of only negative
	edges.
Given a positive seed set $P_0$ and a negative seed set $N_0$,
	define $d_{G^+}(P_0,v)$ to be the shortest graph distance from
	any node in $P_0$ to $v$ only through the positive edges, and
	$d_{G^-}(N_0,v)$ to be the shortest graph distance from
	any node in $N_0$ to $v$ only through the negative edges.
The above distance could be $\infty$ if no such path exists.
Then in the random live-path graph, we say a node $v$ is {\em +active}
	if $d_{G^+}(P_0,v) < \infty$ and $d_{G^+}(P_0,v) < d_{G^-}(N_0,v)$, and
	$v$ is {\em -active} if $d_{G^-}(N_0,v) < \infty$ and $d_{G^-}(N_0,v) \le
	d_{G^+}(P_0,v)$.
The following lemma shows that the positive and negative activation sets
	generated by the above random process is equivalent to the corresponding
	one generated by the CLT model.

\begin{lemma} \label{lem:equirand}
For a given positive seed set $P_{0}$ and negative seed set
$N_{0}$, the distribution over \emph{+active} sets and
\emph{-active} sets is identical in the following two definitions.
\vspace{-2mm}
\begin{enumerate}
\setlength{\itemsep}{-1ex}
\item distribution obtained by running CLT process,
\item
distribution obtained from reachability defined above in the
	live-path graph.
\end{enumerate}
\end{lemma}
\begin{proof}
The activation process under the CLT model consists of several
iterations. In each iteration, some nodes change from
\emph{inactive} to \emph{+active} or \emph{-active}. Thus we first
define $A_{t}^{+}$ to be the set of \emph{+active} nodes at the
end of iteration $t$ and $A_{t}^{-}$ as the set of \emph{-active}
nodes at the end of iteration $t$, for $t=0,1,2...$. Here we consider
a node $v$ which has not been activated by the end of iteration
$t$, namely $v\not\in A_{t}^{+}\cup A_{t}^{-}$. Thus
the probability $v$ becomes \emph{+active} in iteration $t+1$
equals to the chance that the positive influence weights in
$A_{t}^{+}\backslash A_{t-1}^{+}$ push it over the positive
threshold while the negative influence weights is still less than
the negative threshold. The above probability under the
condition that neither node $v$'s negative nor positive threshold
is exceeded already by step $t$ is:
\begin{displaymath}
\frac{(\sum_{u\in A_{t}^{+}\backslash
A_{t-1}^{+}}w^{+}_{uv})(1-\sum_{u\in A_{t}^{-}\backslash
	A_{t-1}^{-}}w^{-}_{uv})}{(1-\sum_{u\in
	A_{t-1}^{+}}w^{+}_{uv})(1-\sum_{u\in A_{t-1}^{-}}w^{-}_{uv})}.
\end{displaymath}
Similarly we can get the probability that node $v$ becomes
\mbox{\emph{-active}} in iteration $t+1$ given than node $v$ is
\emph{inactive} from iteration $0$ to $t$. The probability is:
\begin{displaymath}
\frac{\sum_{u\in A_{t}^{-}\backslash
A_{t-1}^{-}}w^{-}_{uv}}{(1-\sum_{u\in
A_{t-1}^{+}}w^{+}_{uv})(1-\sum_{u\in A_{t-1}^{-}}w^{-}_{uv})}.
\end{displaymath}

On the other hand, we consider the above discussed probability
when using the random live-path graph.
We start from seed set $P_{0}$ and $N_{0}$ and called
them ${B_{0}^{+}}$ and ${B_{0}^{-}}$, respectively.
For each $t=1,2,\ldots$, we define ${B_{t}^{-}}$ to be
	the set containing any $v \not\in {B_{t-1}^{+}} \cup
	{B_{t-1}^{-}}$ such that $v$ has one in-edge from
	some node in ${B_{t-1}^{-}}$; we define
	${B_{t}^{+}}$ to be the set containing any
	$v \not\in {B_{t-1}^{+}} \cup
	{B_{t-1}^{-}}$ such that $v$ has one in-edge from
	some node in ${B_{t-1}^{+}}$ but no in-edge from any node
	in ${B_{t-1}^{-}}$.

By the definition of the random live-path graph, the probability that
	a node $v$ is in $B_{t+1}^+\setminus B_t^+$ conditioned on that $v$ is not
	in ${B_{t}^{+}} \cup {B_{t}^{-}}$ is
\begin{displaymath}
\frac{(\sum_{u\in {B_{t}^{+}}\backslash
{B_{t-1}^{+}}}w^{+}_{uv})(1-\sum_{u\in {B_{t}^{-}}\backslash
{B_{t-1}^{-}}}w^{-}_{uv})}{(1-\sum_{u\in
{B_{t-1}^{+}}}w^{+}_{uv})(1-\sum_{u\in {B_{t-1}^{-}}}w^{-}_{uv})}.
\end{displaymath}
Similarly, the  probability that
	a node $v$ is in $B_{t+1}^-\setminus B_t^-$ conditioned on that $v$ is not
	in ${B_{t}^{+}} \cup {B_{t}^{-}}$ is
\begin{displaymath}
\frac{\sum_{u\in {B_{t}^{-}}\backslash
{B_{t-1}^{-}}}w^{-}_{uv}}{(1-\sum_{u\in
{B_{t-1}^{+}}}w^{+}_{uv})(1-\sum_{u\in {B_{t-1}^{-}}}w^{-}_{uv})}.
\end{displaymath}

The above conditional probabilities are the same as derived from
	the CLT model.
Since $A_{0}^{+}={B_{0}^{+}}$ and $A_{0}^{-}={B_{0}^{-}}$, by
induction over the iterations, we reach at the conclusion that the random
live-path graph model produces the same distribution over
\emph{+active} and \emph{-active} sets as the CLT model.
\hfill $\Box$
\end{proof}

With the equivalence shown in Lemma~\ref{lem:equirand}, we now focus
	on showing the monotonicity and submodularity of negative influence
	reduction in the random live-path graph model.
With a bit of abuse in notation, given a live-path graph $G_X$ and
	a negative seed set $N_0$, we also use $\IBS(S)$ to denote the
	set of nodes in $V$ which would be -active if the positive seed set is
	empty but is not -active if the positive seed set is $S$.
Then the negative influence reduction $\sigma_{\NIR}(S) = E_{G_X}(|\IBS(S)|)$.

Given a set $S$ and a node $u\not\in S$, we say that
	there is a unique path from $S$ to $u$ if there exists some path from
	a node in $S$ to $u$, and
	for any two paths from any two nodes in $S$ to $u$, one path
	must be a sub-path of the other.
In addition, whenever we refer to the unique path
	from $S$ to $u$, we mean the unique shortest path from any node in $S$
	to $u$.
The following lemma shows a simple yet important property of the
live-path graph that leads to the submodularity proof.

\begin{lemma} \label{lem:OnePath}
In a live-path graph $G_X$, for any node $v$, there is a unique
	positive path from some node in the positive
	seed set $S$ to $v$, if $d_{G^+}(S,v)<\infty$,
	and there is a unique negative path from some node
	in the negative seed set $N_0$ to $v$, if $d_{G^-}(N_0,v)<\infty$.
\end{lemma}

\begin{proof}
This is obvious because each node has at most one positive
	in-edge and one negative in-edge.
\hfill $\Box$
\end{proof}

Then we use next two lemmas to give the sufficient and necessary conditions for
$v\in \IBS(S)$ and $v\in \IBS(T\cup\{u\}) \verb"\" \IBS(T)$ in a live-path graph $G_X$.
\begin{lemma} \label{lem:SNCforIBS}
The sufficient and necessary condition for $v\in \IBS(S)$ is:
\begin{enumerate}
\setlength{\itemsep}{-1ex}
\item There exist a unique negative path in $G^{-}$ from node set $N_{0}$ to $v$,
namely $d_{G^{-}}(N_{0},v)<\infty$, and
\item there exists at least one node $u$ in the unique negative
path, such that $d_{G^{+}}(S,u) <
d_{G^{-}}(N_{0},u)$.
\end{enumerate}
\end{lemma}

\begin{proof}
Lemma~\ref{lem:SNCforIBS} is an obvious derivation from the definition of $\IBS(S)$ and Lemma~\ref{lem:OnePath}.
\hfill $\Box$
\end{proof}

\begin{lemma} \label{lem:SNCforIBSE}
The sufficient and necessary
condition for $u\in \IBS(T\cup\{v\}) \verb"\" \IBS(T)$ is:
\begin{enumerate}
\setlength{\itemsep}{-1ex}
\item There exists a unique negative path from $N_{0}$ to $u$,
\item there exists at least one node $w$ on the unique negative path
	from $N_0$ to $u$, such that $d_{G^{+}}(T\cup\{v\},w)<d_{G^{-}}(N_{0},w)$, and
\item for all node $t$ on the unique negative path from $N_0$ to $u$,
	there holds that $d_{G^{+}}(T,t)\geq d_{G^{-}}(N_{0},t)$.
\end{enumerate}
\end{lemma}
\begin{proof}
The above conditions 1 and 2 are direct conclusions of Lemma~\ref{lem:SNCforIBS} on $u\in \IBS(T\cup\{v\})$. Condition 3 is the direct derivation
of Lemma~\ref{lem:SNCforIBS} on $u\not\in \IBS(T)$.\hfill $\Box$
\end{proof}

\begin{lemma} \label{lem:GXsub}
The cardinality set function $|\IBS(S)|$ for a live-path graph
	$G_X$ is monotone and submodular.
\end{lemma}

\begin{proof}
We first prove the monotonicity
of $|\IBS(S)|$, namely for any node $u\in V\backslash (S\cup N_{0})$ and subset $S\subseteq V$, $|\IBS(S)|\leq |\IBS(S\cup\{u\})|$.
We prove the result by showing that $\IBS(S)\subseteq \IBS(S\cup\{u\})$.
Consider any node $v\in \IBS(S)$.
By Lemma~\ref{lem:SNCforIBS}, we have $d_{G^{-}}(N_{0},v)<\infty$, and
	there exists a node $w$ in the unique negative path from $N_0$ to $v$ such
	that $d_{G^{+}}(S,w)<d_{G^{-}}(N_{0},w)$.
It is also clear that $d_{G^{+}}(S\cup\{u\},w)\leq d_{G^{+}}(S,w)$.
Thus, we have $d_{G^{+}}(S\cup\{u\},w) <d_{G^{-}}(N_{0},w)$, and
	by Lemma~\ref{lem:SNCforIBS}, $v\in \IBS(S\cup\{u\})$.

We then prove submodularity of $|\IBS(S)|$ by showing: For any subset $S\subseteq V,T\subseteq V,S\subseteq T$ and $v\in
V\backslash (T\cup N_{0})$,
\begin{displaymath}
\IBS(T\cup\{v\}) \verb"\" \IBS(T) \subseteq \IBS(S\cup\{v\}) \verb"\"
\IBS(S).
\end{displaymath}
Given any $u \in \IBS(T\cup\{v\}) \verb"\" \IBS(T)$,
	we prove that
	$u\in \IBS(S\cup\{v\}) \verb"\" \IBS(S)$ by showing all
three conditions in Lemma~\ref{lem:SNCforIBSE} are satisfied.
The satisfaction of 1 is
obvious, since $d_{G^{-}}(N_0,u)$ doesn't change. As for condition
2, we know that there exists a node $w$ on the unique negative path from $N_{0}$
	to $u$,
$d_{G^{+}}(T\cup\{v\},w)<d_{G^{-}}(N_{0},w)$ and for all node $t$ on
path from $N_{0}$ to $u$, $d_{G^{+}}(T,t)\geq d_{G^{-}}(N_{0},t)$.
Then for node $w$, $d_{G^{+}}(T\cup\{v\},w)<d_{G^{-}}(N_{0},w)\leq
d_{G^{+}}(T,w)$, which implies that
$d_{G^{+}}(T\cup\{v\},w)=d_{G^{+}}(v,w)$. According to Lemma~\ref{lem:OnePath}, the positive influence can
reach node $w$ only in the unique positive path from $v$ to $w$. Thus
$d_{G^{+}}(S\cup\{v\},w)=d_{G^{+}}(v,w)=d_{G^{+}}(T\cup\{v\},w)<d_{G^{-}}(N_{0},w)$.
Then consider condition 3. For any node $t$ in the unique
	negative path from $N_{0}$ to $u$,
	$d_{G^{+}}(T,t)\geq d_{G^{-}}(N_{0},t)$.
Since $S\subseteq T$, it is easy to verify that
	$d_{G^{+}}(S,t) \ge d_{G^{+}}(T,t)$.
Therefore, $d_{G^{+}}(S,t) \ge d_{G^{-}}(N_{0},t)$ and condition 3 also holds.
\hfill $\Box$
\end{proof}

\begin{theorem}
For the CLT model, $\sigma_{\NIR}(S)$ is monotone
and submodular.
\end{theorem}
\begin{proof}
By Lemma~\ref{lem:equirand}, we know that the CLT model is equivalent
	to the random live-path graph model.
By Lemma~\ref{lem:GXsub}, we know that for each live-path graph, the
	size of the influence blocking set is monotone and submodular.
Since $\sigma_{\NIR}(S) = E_{G_X}(|\IBS(S)|)$ and any convex combinations of
	monotone and submodular functions are still monotone and submodular,
	we know that $\sigma_{\NIR}(S)$ is monotone and submodular.
\hfill $\Box$
\end{proof}

\vspace{-2mm}
We have shown that the influence blocking maximization problem
under CLT model is monotone and submodular.
Moreover, we have $\sigma_{\NIR}(\emptyset)=0$.
Then by the famous result in~\cite{NWF78}, the greedy algorithm given in
	Algorithm~\ref{alg:greedy} achieves $1-1/e$ approximation of the optimal
	solution.
The algorithm simply selects seed nodes one by one, and each time
	it always selects the node that provides the largest marginal gain
	to the negative influence reduction.
\begin{algorithm}[t]
\caption{Greedy($k$,$N_{0}$)}
\begin{algorithmic}[1]\label{alg:greedy}
\STATE initialize $S = \emptyset$ \FOR {$i=1$ to $k$} \STATE
select $u= arg\max_{v\in V\backslash (N_{0}\cup S)}(\sigma_{\NIR}(S\cup\{v\}))$ \STATE
$S=S\cup\{u\}$ \ENDFOR
\STATE {\bf return} $S$
\end{algorithmic}
\end{algorithm}

However, the greedy algorithm requires the evaluation of $\sigma_{\NIR}(S)$,
	which cannot be done efficiently.
The standard way of using Monte-Carlo simulations to estimate $\sigma_{\NIR}(S)$
	is slow, especially when we need to simulate the interfering propagation
	of competing influences.
Even with powerful optimization method such as the lazy forward optimization
	of~\cite{JL07} or more advanced approach in \cite{CWY09}, greedy
	algorithm still takes unacceptable long time for large graphs of
	more than $10k$ nodes.
We address this efficiency issue in the next section with our new algorithm
	CLDAG.

\section{CLDAG Algorithm for the IBM Problem} \label{sec:algo}

Motivated by the extremely low efficiency of greedy algorithm, we try to
	tackle this problem with an innovative heuristic
	approach proposed by Chen et al. in \cite{CWW10,CYZ10}.
This heuristic is characterized (a) by restricting influence computation
	of a node $v$ to its local area to reduce computation cost; and
	(b) by carefully selecting a local graph structure for $v$ to allow
	efficient and accurate influence computation for $v$ under this structure.
For the LT model, Chen et al. use a local directed acyclic graph (LDAG)
	structure~\cite{CYZ10}, because it allows linear computation of
	influence in a LDAG, as well as efficient construction of LDAGs using
	an algorithm similar in style to the Dijkstra's shortest path algorithm.
We repeat the LDAG construction algorithm of~\cite{CYZ10} in our
	Algorithm~\ref{alg:findldag} for completeness.
We use $N_{in}(x)$ to denote the set of in-neighbors of node $x$.
The $\theta$ in the algorithm is a threshold from $0$ to $1$ controlling
	the size of the LDAG --- the smaller the $\theta$, the larger the
	LDAG.
The algorithm includes a node $x$ only if its influence to $v$ through
	the LDAG edges are at least $\theta$.
The key update step in line~\ref{alg:linupdate} is based on the important
	linear relationship of activation probabilities in DAG structures
	shown in~\cite{CYZ10}, and repeated below:
\begin{equation}
\ap(x) = \sum_{u\in N_{in}(x)} w_{ux} \cdot \ap(u), \label{eq:linearap}
\end{equation}

\noindent where $\ap(x)$ is the activation probability of node $x$ when
	a seed set is fixed.

\begin{algorithm}[t]
\caption{Find-LDAG($G$,$v$,$\theta$),compute LDAG for $v$ with
threshold $\theta$}
\begin{algorithmic}[1]\label{alg:findldag}
\STATE $X=\emptyset$;$Y=\emptyset$;$\forall v\in V,
\Inf(u,v)=0$;$\Inf(v,v)=1$
\WHILE{$\max_{v\in
V\setminus X}\Inf(u,v)\geq\theta$}
\STATE $x=arg\max_{u\in
V\setminus X}\Inf(u,v)$
\STATE $Y=Y\cup\{(x,u)|u\in X\}$
\STATE $X=X\cup\{x\}$
\FOR{each node $u\in N_{in}(x)$}
\STATE \label{alg:linupdate} $\Inf(u,v)\pequal w_{ux} *\Inf(x,v)$
\ENDFOR \ENDWHILE
\STATE \textbf{return} {$D=(X,Y,w)$ as the \LDAG($v$,$\theta$)}
\end{algorithmic}
\end{algorithm}

However, for the CLT model, negative and positive influence are propagated
	concurrently in the network and interfere with each other.
Thus we need to adjust our LDAG construction and influence computation
	for the CLT model.
First, for each node $v$, we use Algorithm~\ref{alg:findldag} to
	construct two LDAGs, $\LDAG^{+}(v)$ and
	$\LDAG^{-}(v)$, using positive weights and negative weights respectively.
Second, we need to carefully compute the positive activation probability
	$\ap^+(v)$ and negative activation probability $\ap^-(v)$, for
	any node $v$ under the CLT model, assuming positive and negative
	influence are propagated through $\LDAG^{+}(v)$ and
	$\LDAG^{-}(v)$ respectively.
This involves a dynamic programming formulation detailed in the following
	subsection.

\subsection{Influence computation.}
We propose a
dynamic programming method, Inf-CLDAG, to compute the exact activation probability of the central node
$v$ in local structure $\LDAG^{+}(v)$ and
$\LDAG^{-}(v)$.
Under the CLT model, two opposite influence diffusions correlate together
	when disseminating in
	the graph,
	which makes it more tricky than the computation in the origin LT model.
In this case, number of steps taken to activate a node becomes an
	important factor that
	must be taken into consideration when computing the cascade
	result.

For the following computation, we assume that the positive seed set $S$ and
	the negative seed set $N_0$ are fixed, and influence to $v$ only
	diffuses in $\LDAG^{+}(v)$ and $\LDAG^{-}(v)$.
For the IBM problem, we want to compute the negative influence reduction
	under the positive seed set $S$.
It is essentially a computation of negative influence coverage, which
	is given by $\sum_v \ap^-(v)$.

Let $P^{+}(v,t)$ be the probability that the summation of the positive
	weights of in-edges of
	positively activated neighbors of node $v$ exceeds its positive
	threshold exactly at time $t$, and similar for $P^{-}(v,t)$.
Let $\ap^{+}(v,t)$ be the probability that $v$ becomes positively activated
exactly at time $t$,  and similar for $\ap^{-}(v,t)$.
Then we have $\ap^+(v) = \sum_t \ap^{+}(v,t)$ and
	$\ap^-(v) = \sum_t \ap^{-}(v,t)$.
We now show how to compute $\ap^{+}(v,t)$ and $\ap^{-}(v,t)$.

By the definition of the CLT model, we have the following for
	any $v\in V\setminus (S \cup N_0)$ and any $t\ge 1$:
\begin{eqnarray}
& P^{+}(v,t)=\sum_{u\in \LDAG^+(v)}w^+_{uv}\ap^{+}(u,t-1), & \label{equ:PP}\\
&P^{-}(v,t)=\sum_{u\in \LDAG^-(v)}w^-_{uv}\ap^{-}(u,t-1), & \label{equ:PN} \\
& \ap^{+}(v,t)=P^{+}(v,t)(1-\sum_{k=0}^{t}P^{-}(v,k)), & \label{equ:APTP} \\
& \ap^{-}(v,t)=P^{-}(v,t)(1-\sum_{k=0}^{t-1}P^{+}(v,k)).& \label{equ:APTN}
\end{eqnarray}

Equations~(\ref{equ:PP}) and~(\ref{equ:PN}) can be reached by subtracting the probability that the summation of the weights of in-edges of
activated neighbors of node $v$ exceeds
threshold in any round from $0$ to $t-1$ from the corresponding probability for rounds from $0$ to $t$.
Equation~(\ref{equ:APTP}) is derived from the fact that if a node $v$ becomes
positively activated at round $t$, then exactly at round $t$ the
summation of positive weights must exceed the positive threshold, while
	by round $t$ the summation of negative weights does not exceed
	the negative threshold (otherwise $v$ would be negatively activated).
The case for Equation~(\ref{equ:APTN}) is similar.

The boundary conditions of the above equations are
	(a) for $v\in S$, $\ap^+(v,0)=1$,$P^+(v,0)=0$, $P^+(v,t)=\ap^+(v,t)=0$
	for all $t\ge 1$, $P^-(v,t)=\ap^-(v,t)=0$ for all $t\ge 0$;
	(b) for $v\in N_0$, $\ap^-(v,0)=1$,$P^-(v,0)=0$ $P^-(v,t)=\ap^-(v,t)=0$
	for all $t\ge 1$, $P^+(v,t)=\ap^+(v,t)=0$ for all $t\ge 0$; and
	(c) for $v\not\in S\cup N_0$,  $P^+(v,0)=\ap^+(v,0)=
	P^-(v,0)=\ap^-(v,0)=0$.
From the above equations together with the boundary conditions,
 the dynamic programming
	algorithm can be applied to compute the exact activation probability for
	every node $v$.
However, the naive implementation will take $O(m_D \ell_D)$ time, where
	$m_D$ is the size of $\LDAG^+(v)$ and $\LDAG^-(v)$ and $\ell_D$ is
	the length of the longest path in $\LDAG^+(v)$ and $\LDAG^-(v)$.
With a careful planning, as described below,
	we could reduce the time to $O(m_D)$ instead.

Algorithm~\ref{alg:infcomp} provides the pseudocode for our algorithm
	Inf-CLDAG, which computes the negative influence $\ap^-(v)$ to $v$
	from positive seed set $S$ and negative seed set $N_0$, through
	$v$'s LDAGs $\LDAG^{+}(v)$ and $\LDAG^{-}(v)$.
The key feature of the algorithm is the alternating breadth-first-search (BFS)
	traversal on $\LDAG^{-}(v)$ and $\LDAG^{+}(v)$.
Starting from the negative seed set we do one step BFS in $\LDAG^{-}(v)$ and
	compute $P^{-}(x,1)$'s and $\ap^{-}(x,1)$'s for those traversed nodes.
We then do one step BFS in $\LDAG^{+}(v)$ from the positive seeds, and compute
	$P^{+}(x,1)$'s and $\ap^{+}(x,1)$'s for the traversed nodes.
We then go back to $\LDAG^{-}(v)$ to do one more layer of BFS and then
	go back to $\LDAG^{+}(v)$ for one more layer of BFS, and so on.
With this setup, we only need one BFS traversal of $\LDAG^{+}(v)$ and
	$\LDAG^{-}(v)$ to compute all $\ap^-(u,t)$'s, and thus save the running
	time to $O(m_D)$.

\begin{algorithm}[t]
\caption{Inf-CLDAG($v,\LDAG^{+}(v),\LDAG^{-}(v),S,N_{0}$)}
\begin{algorithmic}[1]\label{alg:infcomp}
\STATE $Q^{+}_{0} := S\cap V(\LDAG^{+}(v))$
\STATE $Q^{-}_{0} := N_0 \cap V(\LDAG^{-}(v))$
\STATE initialize $\ap^{+}(u,t),\ap^{-}(u,t),P^{+}(u,t),P^{-}(u,t)$ for
	all $u$ and $t$ to $0$ or according to the boundary condition\\ // can do
	initialization just when needed, so no extra time needed
\STATE set $t=0$

\WHILE{$Q^{+}_{t}\neq\emptyset$ or $Q^{-}_{t}\neq\emptyset$}
\FORALL{node $u$ in $Q^{-}_{t}$}
\FORALL{node $x$ in $\LDAG^{-}(v)$ and $w^{-}_{ux}\neq0$  and $x\not\in S\cup N_{0}$}
\STATE add node $x$ into $Q^{-}_{t+1}$
\STATE $P^{-}(x,t+1)=P^{-}(x,t+1)+w^{-}_{ux}\ap^{-}(u,t)$
\ENDFOR
\ENDFOR
\FORALL{node $x$ in $Q^{-}_{t+1}$}
\STATE $\ap^{-}(x,t+1)=P^{-}(x,t+1)(1-\sum_{k=0}^{t}P^{+}(x,k))$
\ENDFOR

\FORALL{node $u$ in $Q^{+}_{t}$}
\FORALL{node $x$ in $\LDAG^{+}(v)$ and $w^{+}_{ux}\neq0$ and $x\not\in S\cup N_{0}$}
\STATE add node $x$ into $Q^{+}_{t+1}$
\STATE $P^{+}(x,t+1)=P^{+}(x,t+1)+w^{+}_{ux}\ap^{+}(u,t)$
\ENDFOR
\ENDFOR
\FORALL{node $x$ in $Q^{+}_{t+1}$}
\STATE $\ap^{+}(x,t+1)=P^{+}(x,t+1)(1-\sum_{k=0}^{t+1}P^{-}(x,k))$
\ENDFOR
\STATE set $t=t+1$
\ENDWHILE

\STATE $\ap^{-}(v)=\sum_{t}\ap^{-}(v,t)$
\STATE \textbf{return} $\ap^{-}(v)$
\end{algorithmic}
\end{algorithm}

\begin{figure}[!h]
\centering
\includegraphics[width = 3.2in]{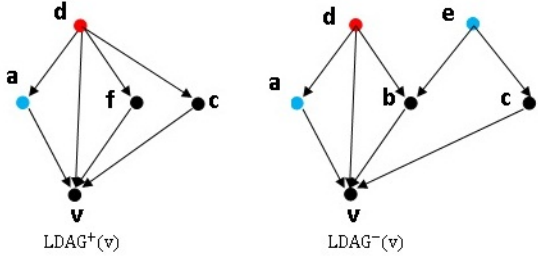}
\caption{A simple example of Inf-CLDAG algorithm (red node $d$ is the
	positive seed and blue nodes $a$ and $e$ are negative seeds).}
\label{fig:InfExample}
\end{figure}
As an example, we show the computation for the structure of
$\LDAG^{+}(v)$ and $\LDAG^{-}(v)$ of Figure~\ref{fig:InfExample}.
In the example, $d$ is the only positive seed while $a$ and $e$ are
	two negative seeds.
In initialization, $\ap^{+}(d,0)$, $\ap^{-}(a,0)$ and $\ap^{-}(e,0)$
	are set to $1$ and all other values are set to $0$.
In the first iteration, we start from the negative seeds $a$ and $e$
	to do one level BFS traversal in $\LDAG^{-}(v)$, and thus
	compute $P^{-}(b,1)$,$P^{-}(c,1)$, $P^{-}(v,1)$,
	$\ap^{-}(b,1)$,$\ap^{-}(c,1)$ and $\ap^{-}(v,1)$.
Next we go to $\LDAG^{+}(v)$ and do one level BFS traversal starting
	from the positive seed $d$, and compute
	$\ap^{+}(f,1)$,$\ap^{+}(c,1)$ and $\ap^{+}(v,1)$, which
	use the values $P^{-}(c,1)$ and $P^{-}(v,1)$ computed.
Then we start the second iteration, which is second level BFS traversal
	in $\LDAG^{-}(v)$, and this only gives us node $v$, for which
	we compute $\ap^{-}(v,2)$.
We will do another BFS traversal on $\LDAG^{+}(v)$, and then we find that
	the BFS traversal has reached all nodes in both LDAGs and the
	computation finishes.

\subsection{CLDAG algorithm.}
Once we have the computation of negative
	influence reduction for any seed set as
	given in Algorithm~\ref{alg:infcomp},
	we can plug it into the greedy algorithm for positive seed selection.
We call this algorithm CLDAG.
We will present the complete pseudocode description of CLDAG algorithm here as Algorithm ~\ref{alg:CLDAG}.

 \begin{algorithm}[t]
 \caption{CLDAG($G,k,N_{0},\theta$)}
 \begin{algorithmic}[1]\label{alg:CLDAG}
 \STATE $S=\emptyset$ \STATE //Build local structure \FORALL{node
 $v\ \in\ V$}
 \STATE Build $\LDAG^{+}(v),\LDAG^{-}(v),OutLS^{+}(v)$
 with threshold $\theta$ \ENDFOR \STATE
 //initialize DecInf
 \STATE set $DecInf(v)=0$ for all node $v\in V$
 \FORALL{node $v\ \in\
 V$ and $v\not\in N_{0}$}
 \FORALL{node $u\in \LDAG^{+}(v)$} \STATE
 $before=\ap^{-}(v,\emptyset)$ \STATE
 $after=\ap^{-}(v,\{u\})$ \STATE
 $DecInf(u) \pequal before-after$ \ENDFOR \ENDFOR \STATE //Main Loop
 \FOR{$i=1$ to $k$}
 \STATE $s=\arg\max_{v\in V\setminus (S\cup N_{0})}DecInf(v)$ \STATE
 //update influence reduction \FORALL{node $v\in OutLS^{+}(s)$}
 \FORALL{node $u\in \LDAG^{+}(v)$} \STATE //subtract previous
 incremental influence reduction \STATE
 $DecInf(u) \mequal \ap^{-}(v,S\cup\{u\})-\ap^{-}(v,S)$
 \STATE
 //add up new incremental influence reduction \STATE $DecInf(u) \pequal
 	\ap^{-}(v,S\cup\{u,s\})- \ap^{-}(v,S\cup\{s\})$
 \ENDFOR \ENDFOR \STATE
 //add node $s$ as positive seed \STATE $S=S\cup\{s\}$ \ENDFOR \STATE\textbf{return} $S$
 \end{algorithmic}
 \end{algorithm}

The algorithm contains an initialization part and an iteration part.
In initialization(line 3-5), we construct $\LDAG^{+}(v)$ and $\LDAG^{-}(v)$ for all
	nodes $v$.
We also maintain an auxiliary set $OutLS^{+}(v)$, which is the
	set of nodes to which $v$ may have positive influence, i.e.,
	$u\in OutLS^{+}(v)$ if and only if $v\in \LDAG^{+}(u)$.
Since positive seed set is changing in the algorithm,
we use $\ap^{-}(v,S)$ to represent the negative activation probability of
 	$v$ in its LDAGs under positive seed set $S$.
Then, for each node $u \in \LDAG^{+}(v)$, we compute the incremental
	influence reduction $\ap^{-}(v,\emptyset)- \ap^{-}(v,\{u\})$
	when adding $u\in \LDAG^{+}(v)$ as a positive seed,
	and sum them up for each $u$ to get $DecInf(u)$, the overall
	incremental influence reduction of node $u$.

In the main iteration(line 16-29), we iterate $k$ times to select $k$ seeds.
In each iteration, we select a new seed $s$ with the largest $DecInf(s)$.
Once $s$ is selected, other nodes' $DecInf(u)$ may need to be updated.
Since $s$ may positively influence all nodes in $OutLS^{+}(s)$, thus
	all nodes $u\in \LDAG^+(v)$ with $v\in OutLS^{+}(s)$ needs to update
	their $DecInf(u)$.
Note that here we take advantage of the local DAG structure, so that
	we do not need to update the incremental influence reduction of
	every node in the graph.
The update is done by using Algorithm~\ref{alg:infcomp}.

\noindent
\textbf{Complexity Analysis.} Let $n=|V|$,
$m_{i\theta}^{+}=\max_{v}|\LDAG^{+}(v)|$, $m_{i\theta}^{-}=\max_{v}|\LDAG^{-}(v)|$,
	and $n_{o\theta}^{+} = \max_{v}|OutLS^{+}(v)|$.
Let $t_{i\theta}^+$ and $t_{i\theta}^-$ be the time of efficient construction
	of $\LDAG^{+}(v)$'s and $\LDAG^{-}(v)$'s, respectively.
Note that $m_{i\theta}^{+}=O(t_{i\theta}^+)$ and $m_{i\theta}^{-}=O(t_{i\theta}^-)$,
	and for sparse graphs, efficient Dijkstra shortest path algorithm
	implementation could make $t_{i\theta}^+$ and $t_{i\theta}^-$ close to
	the order of $m_{i\theta}^{+}$ and $m_{i\theta}^{-}$.
We first analyze the complexity of storing all LDAG structures.

In the initialization step, we need to compute $\LDAG^{+}(v)$'s and
	$\LDAG^{-}(v)$'s for all nodes, and thus it takes
	$O(n(t_{i\theta}^++t_{i\theta}^-))$ time.
We use a max-heap structure to store $DecInf(u)$'s, and it takes
	$O(n)$ time to initialize.
The $DecInf(u)$ computation by Algorithm~\ref{alg:infcomp} takes
	$O(n(m_{i\theta}^{+}+m_{i\theta}^{-}))$ time.
Overall, initialization takes $O(n(t_{i\theta}^++t_{i\theta}^-))$ time.

For the iteration step, each iteration needs to update
	$DecInf(u)$'s for at most $n_{o\theta}^{+} m_{i\theta}^{+}$ nodes,
	and each update involves influence computation by
	Algorithm~\ref{alg:infcomp}, which takes $O(m_{i\theta}^{+}+m_{i\theta}^{-})$
	time, plus updating $DecInf(u)$ on the max-heap, which takes
	$O(\log n)$ time.
Therefore, the iteration step takes
	$O(kn_{o\theta}^{+} m_{i\theta}^{+}(m_{i\theta}^{+}+m_{i\theta}^{-}+\log n))$
	time.

Hence the total time complexity of the algorithm is
	$O(n(t_{i\theta}^++t_{i\theta}^-)+
	kn_{o\theta}^{+} m_{i\theta}^{+}(m_{i\theta}^{+}+m_{i\theta}^{-}+\log n))$.

For space complexity, we store all LDAGs and $OutLS^{+}(v)$'s, so
	the space complexity is $O(n(m_{i\theta}^{+}+m_{i\theta}^{-}+n_{o\theta}^{+}))$.
In actual implementations one may not afford to store all the LDAG structures
	(as in our implementation), so an alternative is to store only
	$OutLS^{+}(v)$'s and compute LDAGs whenever needed.
It is easy to see that in this case, the time complexity is
	$O(n(t_{i\theta}^++t_{i\theta}^-)+
	kn_{o\theta}^{+} m_{i\theta}^{+}(t_{i\theta}^++t_{i\theta}^-+\log n))$, which
	is not significantly worse than storing LDAGs, while the space
	complexity is reduced to $O(n n_{o\theta}^{+})$.

\section{Experiments} \label{sec:exp}
To test the efficiency and effectiveness of CLDAG
for influence blocking maximization problem under the CLT model, we conduct
experiments on three real-world datasets as well as synthetic networks.
\subsection{Experiment setting}
The three real-world datasets are mobile network and collaboration
networks.
The mobile network is a graph derived from a partial call detailed record
	(CDR) data of a Chinese
city from China Mobile, the largest mobile communication service
provider in China. In the mobile network, every node corresponds to a
mobile phone user and the edges correspond to their phone calls
between one another. We use the number of calls between
two users as the edge weight and normalize it among all edges
	incident to a node (the edge thus becomes directed with
	asymmetric edge weights).
The NetHEPT and NetPHY are both collaboration networks
	extracted from the e-print arXiv (http://www.arXiv.org).
The former is extracted
from the "High Energy Physics - Theory" section (form 1991 to 2003),
	and the latter is extracted from "Physics"
	section, and both are the same datasets used in~\cite{CWY09}.
The nodes in both networks are authors and an edge between two nodes
	means the two authors coauthored at least one paper.
We use the number of coauthored papers as the edge weight and
	normalize it among all edges incident to a node.
Some basic statistics of these networks are shown in Table~\ref{tb:stat}.

The edge weights described above do not differentiate between positive
	and negative weights yet.
To differentiate them and study the effect of different diffusion strength
	for positive and negative diffusions,
	we introduce positive propagation rate $p^{+}$ and negative
	propagation rate $p^{-}$, both of which are values from $0$ to $1$.
We multiply edge weight with $p^{+}$ and $p^-$ of each edge to obtain its
	positive and negative edge weight, respectively.
The effect is that all positive edge weights of in-edges of a node
	sums up to $p^+$, and thus with probability $1-p^+$ the node
	will not be activated even if all of its in-neighbors are positively
	activated.
The case for $p^-$ is similar.

\begin{table}[t]
\centering
{\caption{\label{tb:stat} Statistics of the three
        real-world networks. }}
\vspace{1mm}
\begin{tabular}{|c|m{0.5in}|m{0.5in}|m{0.5in}|m{0.5in}}
\hline
  Dataset & Mobile & NetHEPT & NetPHY\\\hline\hline
  Node     & 15.5K & 15.2K & 37.1K\\\hline
  Edge     & 37.0K & 58.9K & 231.5K \\\hline
  Average Degree & 4.77 & 7.75 & 12.48\\\hline
\end{tabular}
\end{table}

We compare the performance of the following
algorithm and heuristics:
\begin{itemize}
\setlength{\itemsep}{-1ex}
\item CLDAG: Our CLDAG algorithm with $\theta=0.01$;\footnote{We found
	that $\theta < 0.01$ will not have significant improvement for
	the blocking effect, for all networks tested.}
\item Greedy: Algorithm~\ref{alg:greedy} under the CLT model with the
	lazy-forward optimization of~\cite{JL07}, and 10000 simulation runs for
	each influence estimate.

\item Degree: a baseline heuristic, simply
	choosing nodes with largest degrees as positive seeds.
\item
Random:  a baseline heuristic, simply choosing
	nodes at random as positive seeds.

\item Proximity Heuristic: A simple heuristic under which we
choose the direct out-neighbors of negative seeds as positive seeds to
	block the negative influence.
Among these direct out-neighbors, we sort them by the negative weights of
	their in-edges connecting them with negative seeds,
	and select the top $k$ nodes as
	the positive seeds.
\end{itemize}

Proximity heuristic introduced above is based on the simple idea of
	trying to block the influence of negative seeds at their direct
	neighbors.
It should be noticed that the proximity heuristic can be considered as
a simplified version of our CLDAG algorithm.
In fact, for each node $v$, if we construct its $\LDAG^{+}(v)$ to be only
	the node $v$ itself, while its $\LDAG^{-}(v)$ to be $v$ itself if $v$
	has no in-neighbors in the negative seed set $N_0$, or else
	to be $v$ with one of $v$'s in-neighbors in $N_0$ with the largest
	negative edge weight to $v$.
It is easy to verify that our CLDAG algorithm under these LDAG structures
	exactly matches the proximity heuristic.
Therefore, proximity heuristic can be treated as an intermediate algorithm
	between the baseline random algorithm and the full-blown
	CLDAG algorithm, and is
	helpful for understanding the features of CLDAG.

Since the CLT model is a probabilistic model, when
we evaluate the blocking effect for any given positive and negative seed
	sets, we test it for 1000 times
	and take their average as the result. The negative
seeds in $N_0$ are chosen either randomly or 	
from nodes with the largest degrees.
The scalability test is run on Intel Xeon E5504 2G*2 (4 cores for every CPU), 36G memory server, while all others are run on Dell D630 laptop with 2G memory.
All experiment code is written in C++.

\subsection{Results with the greedy algorithm.}

We first run tests that include the greedy algorithm.
Since the greedy algorithm runs very slow on large graphs, we extract
	two subgraphs from the datasets for comparison.
One subgraph is a 1000 node graph extracted from the mobile network, and
	another is a 5000 node graph extracted from the NetHEPT network.
The extraction is done by randomly selecting a node in the graph and doing
	BFS from the node until we obtain the desired number of nodes, and we
	include all edges for these nodes in the subgraph.
We choose 50
nodes with the highest degrees as negative seeds and select 200 positive
seeds to block their influence.
Both $p^{+}$ and $p^{-}$ are set to 1.
The experiment result are showed in Figure~\ref{fig:greedy}.

\begin{figure}[t]
\centering
\begin{tabular}{cc}
\includegraphics[width = 1.6in]{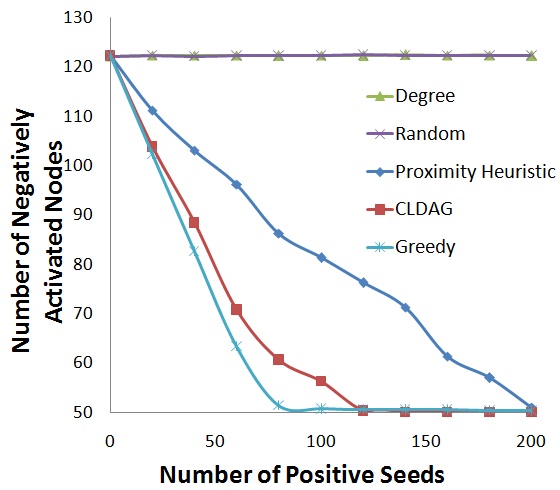} & \includegraphics[width = 1.6in]{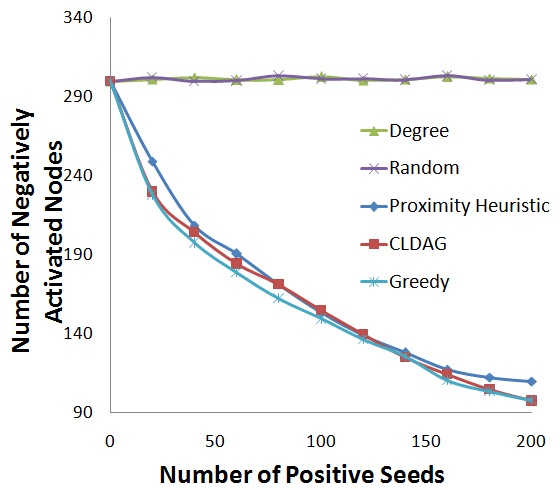}\\
(a)Mobile &
(b)NetHEPT\\
\multicolumn{2}{c}{\includegraphics[width=3.2in]{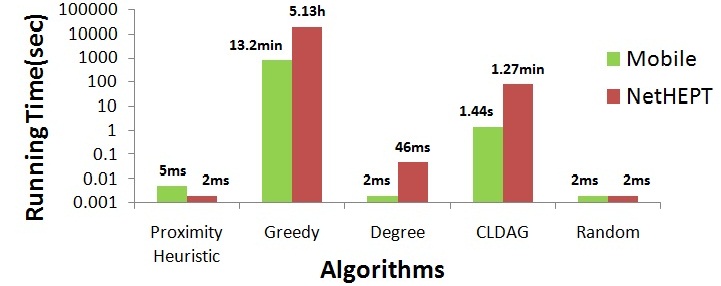}}\\
\multicolumn{2}{c}{(c)Running time for selecting 200 seeds}
\end{tabular}
\caption{Experiment result of comparison with Greedy algorithm.}
\label{fig:greedy}
\end{figure}

From Figure~\ref{fig:greedy} (a) and (b), we can see that
	the CLDAG algorithm consistently
matches the performance of the greedy algorithm for both datasets.
In the 1000-node mobile network test, CLDAG significantly outperforms
	the Proximity heuristic, e.g., when CLDAG completely blocks all
	negative influence with 130 seeds, proximity heuristic still allows
	negative influence to reach about 30 more nodes.
In term of negative influence reduction, this is $(120-50)/(120-80)
	= 175\%$ improvement.
In the 5000-node NetHEPT dataset, proximity heuristic performs as well
	as CLDAG and the greedy algorithm.
In both cases, random and degree heuristic perform badly, essentially
	having no blocking effect at all.
This is in contrast with degree heuristic result for influence maximization
	reported in the previous papers~\cite{CWY09,CWW10,CYZ10}, where degree
	heuristic still have moderate gain when selecting more seeds.
Our interpretation is that for influence blocking maximization, knowing where
	the negative seeds are becomes very important, and thus proximity heuristic
	could behave reasonably well while degree heuristic oblivious to the
	location of negative seeds becomes useless.

From Figure~\ref{fig:greedy} (c), we see that CLDAG is much
	faster than the greedy algorithm, with more than two orders of magnitude
	speedup.
With 5000 nodes, the greedy algorithm already takes more than five hours,
	while CLDAG only takes one minute to select 200 seeds.

We further compare the scalability of CLDAG with the greedy algorithm.
For this test, we use a family of synthetic power-law graphs generated by the
DIGG package \cite{digg}. We generate graphs with doubling number
of nodes, from 0.2K, 0.4K, up to 6.4K, using power-law exponent of
2.16. Each size has 10 different random graphs and our running
time result is the average among the runs on these 10 graphs. We
randomly choose 50 nodes as negative seeds and find 50 positive
seeds to block the negative influence. We set both $p^{+}$ and $p^{-}$ to 1.
The scalability result is
shown in Figure~\ref{fig:scalability}.

\begin{figure}[t]
  \centering
    \includegraphics[width=0.36\textwidth]{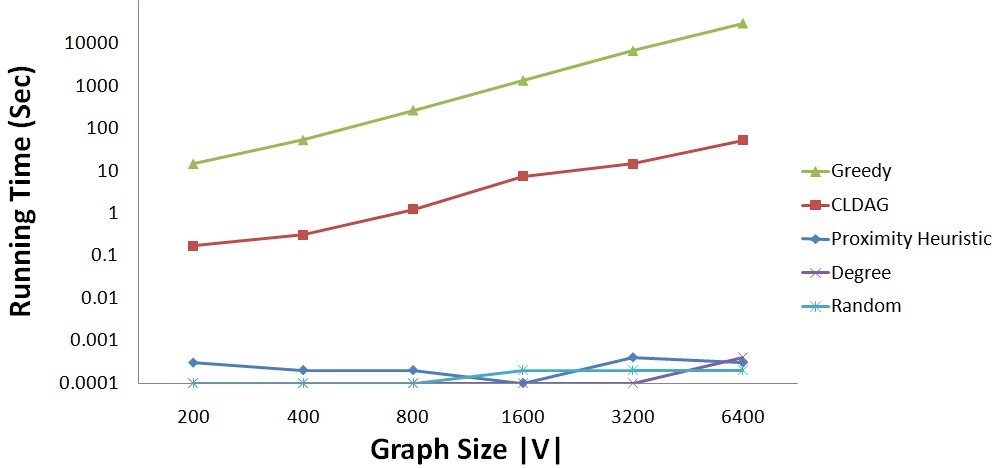}
  \caption{Experiment result on algorithm scalability.}
	\label{fig:scalability}
\end{figure}

The result clearly shows that CLDAG is two orders of magnitude
faster than the greedy algorithm and its running time
	has linear relationship with
	the size of the graph, which indicates
	good scalability of the CLDAG algorithm.
Therefore, comparing with the greedy algorithm, CLDAG matches the
	blocking effect of the greedy algorithm while has at least
	two orders of magnitude speedup in running time.

\subsection{Results on larger dataset without the greedy algorithm.}
We conduct experiments on the full graphs of the three datasets, but
	we do not include the greedy algorithm since its running time becomes
	too slow.
The initial
negative seeds are chosen either randomly or with highest degrees.
	We first set $p^{+}$ and $p^{-}$
to 1.

\begin{figure}[t]
\centering
\begin{tabular}{cc}
\includegraphics[width = 1.6in]{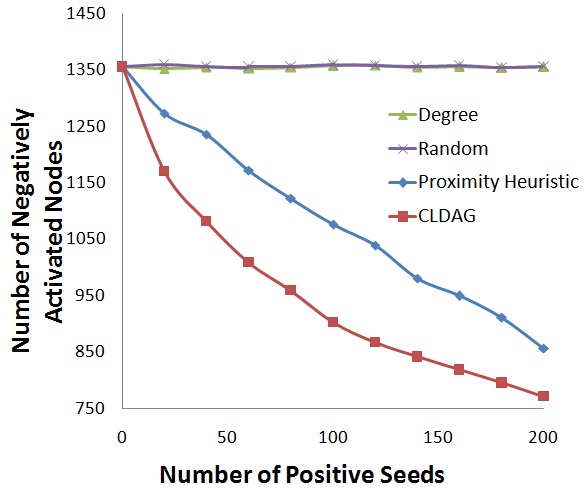}
&
\includegraphics[width = 1.6in]{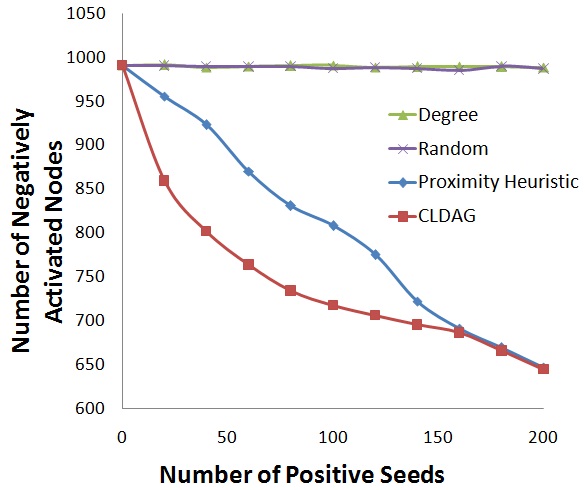}
\\
(a)NetHEPT: Max degree   & (b) NetHEPT:Random  \\
\includegraphics[width = 1.6in]{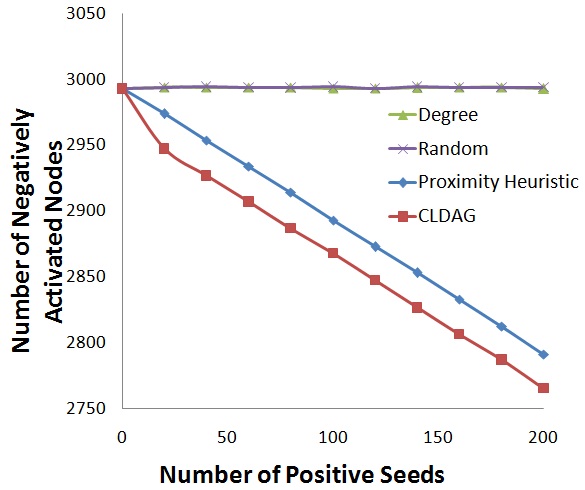}
&
\includegraphics[width = 1.6in]{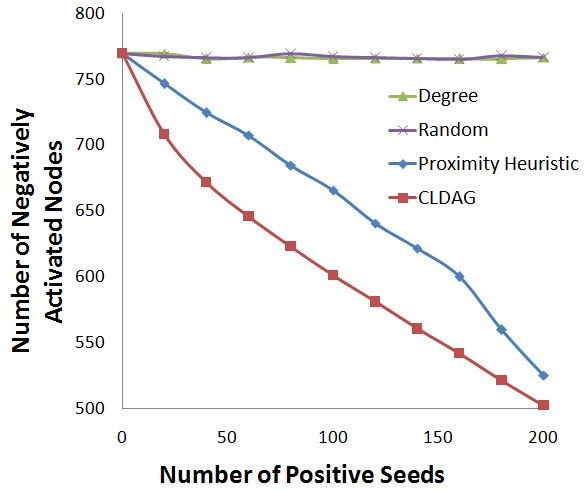}
\\
(c)Mobile: Max degree   & (d) Mobile:Random  \\
\includegraphics[width = 1.6in]{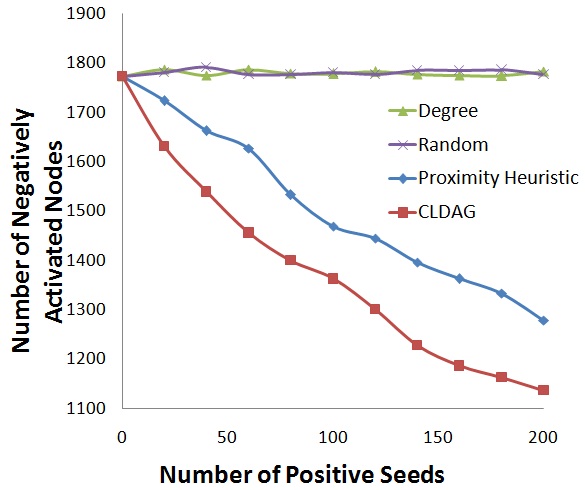}
&
\includegraphics[width = 1.6in]{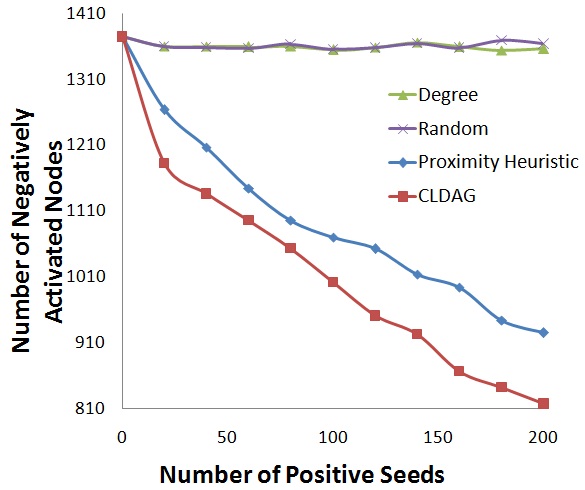}\\
(e)NetPHY: Max degree   & (f) NetPHY:Random  \\
\multicolumn{2}{c} {\includegraphics[width = 3.2in]{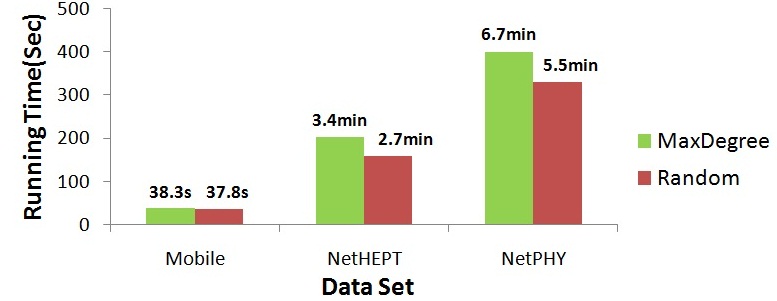}}\\
\multicolumn{2}{c} {(g) Running Time of CLDAG algorithm on real networks}
\end{tabular}
\caption{Experiment result of CLT model on three real dataset.
We choose 200 negative seeds with max degree in experiment (a),(c),(e) and 400 random negative seeds in
experiment (b),(d),(f). }
\label{fig:real}
\end{figure}

As shown in Figure~\ref{fig:real} (a) to (f),
	the performance of CLDAG
	strictly dominates the proximity heuristic in all cases.
For random negative seed selection, the negative influence reduction of CLDAG
	is on average 78.24\% higher than that of the proximity algorithm
	(percentage taken as the average of results from $1$ seed to $200$ seeds).
For max-degree negative seed selection, CLDAG improves the
	performance of proximity heuristic even more, for 80.75\% on
average.
Degree and random heuristic still show no blocking effect on all test cases.
The running time of CLDAG is consistently low, as shown in
	Figure~\ref{fig:real} (g).
The results demonstrate that across all networks and all negative seed
	selection methods, CLDAG has consistently
	good performance in negative influence reduction over other
	heuristics, and it achieves this good performance efficiently.

Next, we vary propagation rate $p^{+}$ and $p^{-}$ to check their effect on influence dissemination and the performance
of our algorithm. For simplicity, we only present experiment result on
	the NetHEPT network. 
We choose 200 nodes with max degree as negative seeds and select 200 positive nodes to block their influence.
In one test we have $p^{+}=0.5$ and $p^{-}=1$, and thus negative influence
	diffusion is stronger, while in the second test,
	we use $p^{+}=1$ and $p^{-}=0.5$, making positive influence diffusion
	stronger.
\begin{figure}[t]
\centering
\begin{tabular}{cc}
\includegraphics[width = 1.6in]{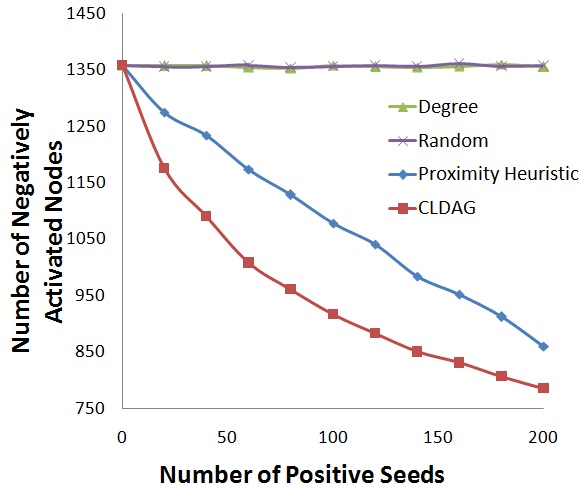}
&
\includegraphics[width = 1.6in]{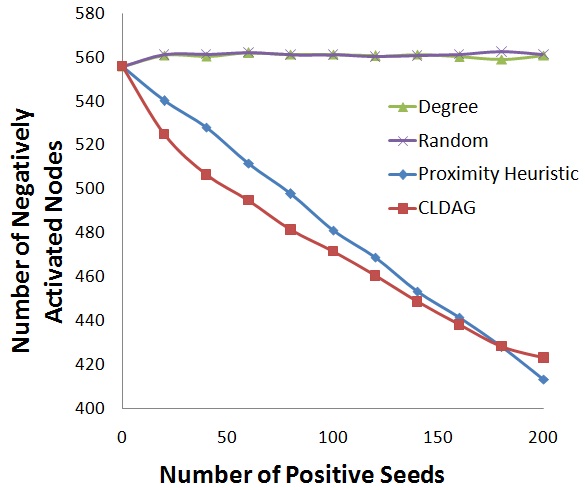}
\\
(a) $p^{+}=0.5$, $p^{-}=1$  & (b) $p^{+}=1$, $p^{-}=0.5$
\end{tabular}
\caption{Experiment result of CLT model on propagation rate $p^{+}$ and $p^{-}$.}
\label{fig:resultP}
\end{figure}

Figure~\ref{fig:resultP} reports our simulation results.
First, as expected, when the negative influence is stronger, more
	nodes become negative without positive influence ($1350$ nodes
	vs. $560$ nodes in our two test cases).
More importantly, we see that our CLDAG algorithm performs much better
	than the proximity heuristic when the negative influence is stronger
	(Figure~\ref{fig:resultP} (a)).
This is because in this case negative diffusion can traverse long paths
	and thus simply placing positive seeds next to
	the negative seeds may not block the negative diffusion well.
On the other hand, when the negative influence is weak
	(Figure~\ref{fig:resultP} (b)), negative influence could be effectively
	blocked by placing positive seeds next to them, and thus
	proximity heuristic performs close to CLDAG.

To summarize, our results show that CLDAG has the best performance
	among tested heuristics across all graphs, and especially when
	negative influence diffusion is strong.
Proximity heuristic as a simplified version of CLDAG has reasonable performance
	in a few cases especially when negative influence diffusion is weak, and
	can be used as a fast alternative to CLDAG in this case.
However, there are situations
	in which proximity heuristic is significantly worse than CLDAG.
Traditional degree heuristic cannot be used for influence blocking maximization
	at all from our test results.

\subsection{Effectiveness of influence blocking at different negative
	seed size.}
Finally, we test the effectiveness of influence blocking with
	CLDAG, when the size of negative seeds increases.
We vary the negative seed size from $1$ to $1000$, and see how many
	positive seeds are required by CLDAG to reduce negative influence
	to $10\%$.
We cap the number of positive seeds at $1000$.
For this test, we use the NetHEPT network, select negative seeds with
	largest degrees, and set $p^+=p^-=1$.
The results are shown in Table~\ref{tb:feasibility}, where
	$\sigma_{N}(S,N_{0})$ denotes the expected number of negative activations
	with positive seeds $S$ and negative seeds $N_0$.
\begin{table}[t]
\caption{\label{tb:feasibility} Result on the effectiveness of influence blocking}
\begin{center}
\begin{tabular}{|c|c|c|c|}
  \hline
  \multicolumn{1}{|c|}{$|N_{0}|$}
  &\multicolumn{1}{c|}{$\sigma_{N}(\emptyset,N_{0})$}
  &\multicolumn{1}{c|}{$|S|$}
  &\multicolumn{1}{c|}{$\sigma_{N}(S,N_{0})$}
\\\hline\hline
  1     &  72.8979 & 23 & 6.7396\\\hline
  2     &  77.4516 & 68 & 6.0182\\\hline
  5     &  156.48 & 145 & 15.6667\\\hline
  10     &  213.077 & 199 & 20.6628\\\hline
  20     &  581.366 & 557 & 57.6617\\\hline
  50     &  963.633 & 926 & 95.8451\\\hline
  100     &  1006.37 & 1000 & 108.823\\\hline
  200     &  1669.85 & 1000 & 680.518\\\hline
  500     &  3635.95  & 1000 & 2640.8\\\hline
  1000     &  5836.48 & 1000 & 4845.58\\\hline
\end{tabular}
\end{center}
\end{table}

The result shows that it requires about $20$ to $30$ times of positive
	seeds to reduce negative influence to about $10\%$ level, and it
	becomes increasingly hard to block negative influence.
For example, with $1000$ negative seeds, we spend an equal number of
	$1000$ positive seeds but can only reduce $17\%$ negative influence.
Therefore, first mover has a clear advantage, and the best way to block
	negative influence is before it becomes pervasive.

\section{Conclusion and Discussions} \label{sec:conclude}
In this work, we study influence blocking maximization problem
	under the competitive linear threshold model.
We show that the objective function of the IBM problem is submodular
	under the CLT model, and thus the greedy approximation algorithm is
	available.
We then design an efficient algorithm CLDAG to overcome the slowness of
	the greedy algorithm.
Our simulation results demonstrate that CLDAG matches the greedy algorithm
	in the blocking effect while significantly improving running time.
CLDAG also outperforms other heuristic algorithms such as proximity
	heuristic that selects direct neighbors of negative seeds, showing that
	CLDAG is a stable and robust algorithm for the IBM problem.

Finally, we compare two closely related results in the literature, which
	showing some interesting subtleties in competitive influence diffusion.
First, in~\cite{WWW11}, Budak et al. study the IBM problem for the
	extended IC model.
They show, however, that when we extend the IC model to allow positive
	 and negative diffusions having two set of different parameters, the IBM
	is not submodular.
This indicates a subtle difference between different diffusion models.
In this sense, CLT model is more expressive, since it is easier to model
	different diffusion strength in the CLT model and see its effect, as
	we did in our evaluation (Figure~\ref{fig:resultP}).
They also show that when restricting the positive weights to be $1$, or
	to be the same as negative weights, the problem becomes submodular.
For these cases, we are able to design efficient algorithms close to
	MIA and MIA-N of~\cite{CWW10, Chen2011}, and our simulations results
	are similar when comparing with the greedy algorithm and other heuristics,
	but we do not report them here.

Second, in~\cite{BFO10}, Borodin et al. propose several competitive
	diffusion models extended from the LT model.
In particular, their separate threshold model is essentially the CLT model
	in this paper (with a slightly different tie-breaking rule).
Interestingly, they show that the problem of maximizing positive influence
	given a fixed negative seed set is {\em not} submodular
	(applicable to our CLT model), while
	we show here that influence blocking maximization {\em is} submodular.
Intuitively, this is because even though a positive seed $x$ blocks the negative
	influence, to maximize positive influence it may also need other positive
	seeds to activate nodes that are blocked from negative influence
	by node $x$.
Therefore, the marginal gain of $x$ is larger for the positive influence
	maximization objective when there are other positive seeds
	corporating with $x$,
	making it not submodular.

Several improvements and future directions are possible.
One direction is looking into even faster and more space-efficient algorithms
	for influence blocking maximization.
Another direction is to tackle the IBM problem in other competitive
	diffusion models, especially models without submodularity property.

\bibliographystyle{abbrv}
\bibliography{SDM11bib}

 \clearpage

 \section*{Appendix}
 \appendix

 \section{Proof of Theorem~\ref{thm:nph}}

 \begin{proof}
 Consider an instance of the NP-complete vertex cover problem defined by an
 	undirected $|V|$-node graph $G=(V,E)$
 and an integer $k$.
 The vertex cover problem asks if there exists a set $S$ of $k$ nodes in $G$ so that every edge has at least one endpoint in $S$.
 We show that this can be reduced to the IBM problem under the CLT model.

 Given an instance of the vertex cover problem involving a graph $G$,
 	let $d_G(v)$ denote the degree of node $v$ in $G$ and define $d_m=\max_{v\in V}d_G(v)$.
 We construct a corresponding instance of the IBM problem as follows.
 We first construct a directed graph $G'$ from $G$.
 Besides the original graph structure,
 	for each vertex $v_i$ in $G$, we
 build a spindle structure $S_i$ with $|V|+2$ nodes and a chain $C_i$
 	with $\ell+1$ nodes, and $S_i$ and $C_i$ share the node $v_i$. We
 use a toy example given in Figure~\ref{fig:NPhard} to describe
 	our construction.

 \begin{figure}[h!]
   \centering
     \includegraphics[width=0.5\textwidth]{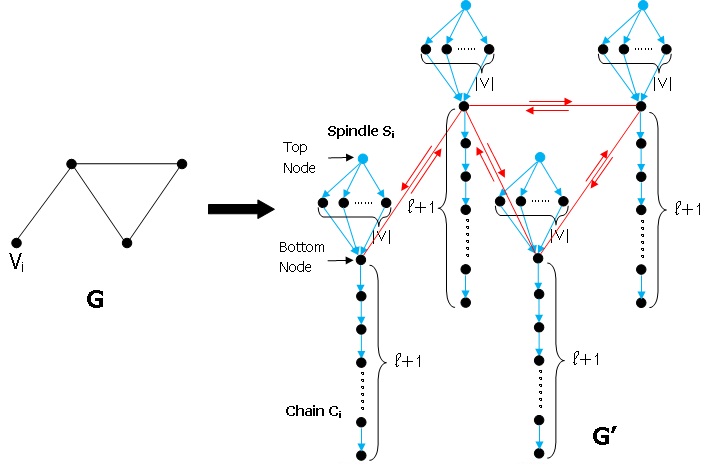}
   \caption{Graph construction for the NP-hardness proof (positive edges
 	are red, and negative seeds and negative edges are blue).}\label{fig:NPhard}
 \end{figure}

 As shown in the Figure~\ref{fig:NPhard},
 	each spindle structure consists of a top node, $|V|$ intermediate
 nodes and a bottom node. The top node is chosen as a negative seed and has $|V|$ negative edges (meaning positive weight is $0$) with
 	negative weight $1$ to each of the
 intermediate nodes. Each intermediate node has a negative edge with
 	$\frac{1}{|V|}$ negative weight to the bottom node.
 Then we use each bottom node of all spindle structures to form a similar graph as $G$ except that
 	we direct all edges of the origin $G$ in both directions to build
 	positive edges (meaning negative weights are $0$).
 The positive weights of positive edges are set according to the
 	degree of the according node in $G$. Namely for the bottom node $v_i$ of
 spindle structure $S_i$, we set all the weights of positive in-edges of
 	$v_i$ equally to
 	$\frac{1}{d_G(v_i)}$.
 Next, starting from $v_i$, we add a chain with $\ell+1$ nodes
 	(including $v_i$) and $\ell$ directed negative edges of weight $1$.
 We set $\ell=\lceil\frac{|V|d_m}{|V|-1}-1\rceil$.
 Thus the total size of constructed graph $G'$ is $O(|V|^2)$.


 We first show Lemma~\ref{lem:NPhard} for our NP-hardness proof.
 \begin{lemma}\label{lem:NPhard}
 In the constructed graph $G'$, given positive seed set $S$ if there exists a bottom node $v$ in a spindle structure whose positive activation probability at step $1$ is not strictly $1$, a higher negative influence reduction with one more positive seed can be achieved by choosing $v$ instead of selecting any other intermediate node in spindle structure or any node in chains.
 \end{lemma}
 \begin{proof}
 We assume that the positive activation probability for node $v$ at step $1$ is $p^{+}$. Firstly, it is obvious that choosing
 bottom node is a better strategy than choosing node in any chain. Then by adding node $v$ to the positive
 seed set, we can have a negative influence reduction $\Delta_v\geq(1-p^{+})(\ell+1)$. By adding any intermediate node to positive seed set, we can have $\Delta_{inter}\leq1+\frac{1}{|V|}(1-p^{+})(\ell+1)$. With $\ell=\lceil\frac{|V|d_m}{|V|-1}-1\rceil$, we can easily get $\Delta_v\geq(1-p^{+})(\ell+1)>1+\frac{1}{|V|}(1-p^{+})(\ell+1)\geq\Delta_{inter}$. Therefore choosing bottom node $v$ will always lead to
 greater gain in negative influence reduction than any other intermediate or chain nodes.
 \hfill $\Box$
 \end{proof}
 If there is a vertex cover $S$ of size $k$ in $G$, then one can deterministically make
 $\sigma_{\NIR}(S)=|V|(\ell+1)$ by choosing the positive seed set as the vertex cover of graph $G$. Since without the positive seeds all nodes in $G'$ will be negatively activated, while with
 positive seed set $S$ we can save the bottom nodes and also the nodes on the chains.
 Conversely this is the only way to get a set
 with $\sigma_{\NIR}(S)\ge |V|(\ell+1)$.
  Otherwise if positive seeds among bottom nodes are not a vertex cover of
 	the origin graph $G$, the probability that all bottom nodes can
 be positively activated in step $1$ is strictly less than $1$, and the
 	gap is at least $1/d_m$.
 According to Lemma~\ref{lem:NPhard}, all $k$ positive seeds must be chosen among the bottom nodes.
 Thus, in step $2$ any node that was not positive in step $1$ must
 	become negative, due to negative influence dominance.
 Hence, we have $\sigma_{\NIR}(S)\leq(|V|-1)(\ell+1)+(1-\frac{1}{d_m})(\ell+1)<|V|(\ell+1)$.
 Therefore, by checking if $G'$ has a positive seed set of size $k$ that
 	achieves negative influence reduction of at least $|V|(\ell+1)$,
 	we can know if the original graph $G$ has a vertex cover of size $k$.
 \hfill $\Box$
 \end{proof}

\end{document}